\documentclass[11pt]{article}

\usepackage[a4paper,left=20mm,top=20mm,right=20mm,bottom=25mm]{geometry}

\usepackage{amsmath, amsfonts, amssymb, amsthm}
\usepackage[noend]{algorithmic}
\usepackage[linesnumbered,ruled,vlined]{algorithm2e}

\usepackage{mathtools}
\usepackage[mathscr]{euscript}
\usepackage{wasysym}

\usepackage{graphicx}
\usepackage{xcolor}
\usepackage{enumitem}
\usepackage{authblk}

\usepackage{color}

\usepackage[colorlinks=true,urlcolor=black,linkcolor=blue,citecolor=blue]{hyperref}
\usepackage[font=footnotesize,width=.85\textwidth,labelfont=bf]{caption}

\newtheorem{theorem}{Theorem}[section]
\newtheorem{lemma}[theorem]{Lemma}
\newtheorem{corollary}[theorem]{Corollary}

\newtheorem{proposition}[theorem]{Proposition}
\newtheorem{fact}[theorem]{Observation}

\usepackage{underscore}
\usepackage{hyperref}

\usepackage[numbers,sort]{natbib}

\usepackage{bigfoot}
\DeclareNewFootnote[para]{default}
\DeclareNewFootnote{A}[arabic]
\DeclareNewFootnote[para]{B}[fnsymbol]
\makeatletter
\def\moverlay{\mathpalette\mov@rlay}
\def\mov@rlay#1#2{\leavevmode\vtop{%
    \baselineskip\z@skip \lineskiplimit-\maxdimen
    \ialign{\hfil$\m@th#1##$\hfil\cr#2\crcr}}}
\newcommand{\charfusion}[3][\mathord]{
  #1{\ifx#1\mathop\vphantom{#2}\fi
    \mathpalette\mov@rlay{#2\cr#3}
  }
  \ifx#1\mathop\expandafter\displaylimits\fi}
\DeclareRobustCommand\bigop[1]{%
  \mathop{\vphantom{\sum}\mathpalette\bigop@{#1}}\slimits@
}
\newcommand{\bigop@}[2]{%
  \vcenter{%
    \sbox\z@{$#1\sum$}%
    \hbox{\resizebox{\ifx#1\displaystyle.9\fi\dimexpr\ht\z@+\dp\z@}{!}{$\m@th#2$}}%
  }%
}
\makeatother

\newcommand{\lca}{\operatorname{lca}}
\newcommand{\parent}{\operatorname{parent}}
\newcommand{\child}{\operatorname{child}}
\newcommand{\HH}{\mathcal{H}}

\DeclareMathOperator{\build}{\mathtt{BUILD}}

\providecommand{\keywords}[1]{\textbf{\textit{Keywords: }} #1}

\title{A Simple Linear-Time Algorithm for the Common Refinement of Rooted 
Phylogenetic 
Trees on a Common Leaf Set}

\author[1]{David Schaller}
\author[2]{Marc Hellmuth}
\author[1,3-7]{Peter F.\ Stadler}

\affil[1]{Bioinformatics Group, Department of Computer Science, and
  Interdisciplinary Center for Bioinformatics, Universit{\"a}t Leipzig,
  H{\"a}rtelstrasse 16-18, D-04107 Leipzig, Germany
  \authorcr \texttt{sdavid@bioinf.uni-leipzig.de} $\cdot$ 
  \texttt{studla@bioinf.uni-leipzig.de}}

\affil[2]{Department of Mathematics, Faculty of Science,
  Stockholm University, SE-10691 Stockholm, Sweden
  \authorcr \texttt{marc.hellmuth@math.su.se}}

\affil[3]{Competence Center for Scalable Data Services and Solutions
Dresden/Leipzig, German Centre for Integrative Biodiversity Research
(iDiv), and Leipzig Research Center for Civilization Diseases, Universit{\"a}t 
Leipzig, Augustusplatz 12, D-04107
Leipzig, Germany.}

\affil[4]{Max Planck Institute for Mathematics in the Sciences,
  Inselstra{\ss}e 22, D-04103 Leipzig, Germany}

\affil[5]{Department of Theoretical Chemistry, University of Vienna 
W{\"a}hringer Stra{\ss}e 17, A-1090 Vienna,Austria}

\affil[6]{Facultad de Ciencias, Universidad Nacional de Colombia, Sede 
Bogot{\'a}, Ciudad Universitaria, COL-111321 Bogot{\'a}, D.C., Colombia}

\affil[7]{Santa Fe Institute, 1399 Hyde Park Rd., NM87501 Santa Fe, USA}

\date{\ }

\begin{document}
  
  \maketitle 
  
  \abstract{
    \textbf{Background} The supertree problem, i.e., the task of finding a
    common refinement of a set of rooted trees is an important topic in
    mathematical phylogenetics. The special case of a common leaf set $L$ is
    known to be solvable in linear time. Existing approaches refine one input
    tree using information of the others and then test whether the results
    are isomorphic.
    
    \textbf{Results} A linear-time algorithm, \texttt{LinCR}, for
    constructing the common refinement $T$ of $k$ input trees with a common
    leaf set is proposed that explicitly computes the parent function of $T$
    in a bottom-up approach.
    
    \textbf{Conclusion} \texttt{LinCR} is simpler to implement than other
    asymptotically optimal algorithms for the problem and outperforms the
    alternatives in empirical comparisons.
    
    \textbf{Availability} 
    An implementation of \texttt{LinCR} in
    Python is freely available at
    \url{https://github.com/david-schaller/tralda}.
  }

  \bigskip
  \noindent
  \keywords{
    mathematical phylogenetics,
    rooted trees,
    compatibility of rooted trees}

\sloppy

\section*{Introduction}

Given a collection of rooted phylogenetic trees $T_1$, $T_2$, \dots $T_k$,
the supertree problem in phylogenetics consists in determining whether
there is a common tree $T$ that ``displays'' all input trees $T_i$,
$1\le i\le k$, and if so, a supertree $T$ is to be constructed
\cite{Sanderson:98,Semple:00}. In its most general form, the leaf sets
$L(T_i)$, representing the taxonomic units (taxa), may differ, and the
supertree $T$ has the leaf set $L(T)=\bigcup_{i=1}^k L(T_i)$. Writing
$n\coloneqq |L(T)|$, $N\coloneqq \sum_{i=1}^k |L(T_i)|$, and
$R\coloneqq \sum_{i=1}^k |L(T_i)|^2$, this problem is solved by the
algorithm of Aho et al.\ \cite{Aho:81}, which is commonly called
\texttt{BUILD} in the the phylogenetic literature \cite{Semple:03}, in
$O(N n)$ time for binary trees and $O(R n)$ time in general.

An $O(N^2)$ algorithm to compute all binary trees compatible with the input
is described in \cite{Constantinescu:95}.  Using sophisticated data
structures, the effort for computing a single supertree was reduced to
$O(\min(N \sqrt{n},N+n^2\log n))$ for binary trees and $(R\log^2 R)$ for
arbitrary input trees \cite{Henzinger:99}. Recently, an $O(N \log^2 N)$
algorithm has become available for the compatibility problem for general
trees \cite{Deng:18}. The compatibility problem for nested taxa in addition
assigns labels to inner vertices and can also be solved in $O(N \log^2 N)$
\cite{Deng:17}.

Here we consider the special case that the input trees share the same leaf
set $L(T_1)=L(T_2)=\dots=L(T_k)=L(T)=L$, and thus $N=kn$ and
$R=kn^2$. While the general supertree problem arises naturally when
attempting to reconcile phylogenetic trees produced in independent studies,
the special case appears in particular when incompletely resolved trees are
produced with different methods. In a recent work, we have shown that such
trees can be inferred e.g.\ as the least resolved trees from best match
data \cite{Geiss:19a,Schaller:21d} and from information of horizontal gene
transfer \cite{Geiss:18a,Hellmuth:19}.  Denoting with $\HH(T)$ the set of
``clusters'' in $T$, we recently showed that the latter type of data can be
explained by a common evolutionary scenario if and only if (1) both the
best match and the horizontal transfer data can be explained by least
resolved trees $T_1$ and $T_2$, respectively, and (2) the union
$\HH(T_1)\cup\HH(T_2)$ is again a hierarchy. In this context it is of
practical interest whether the latter property can be tested efficiently,
and whether the common refinement $T$ satisfying
$\HH(T)=\HH(T_1)\cup\HH(T_2)$ \cite{Hellmuth:21x} can be constructed
efficiently in the positive case.

Several linear time, i.e., $O(|L|)$ time, algorithms for the common
refinement of two input trees $T_1$ and $T_2$ with a common leaf set have
become available. The \texttt{INSERT} algorithm \cite{Warnow:94}, which
makes use of ideas from \cite{Gusfield:91}, inserts the clusters of $T_2$
into $T_1$ and \emph{vice versa} and then uses a linear-time algorithm to
check whether the two edited trees are isomorphic \cite{Aho:74}.  Assuming
that the input trees are already known to be compatible,
\texttt{Merge_Trees} \cite{Jansson:13,Jansson:16} can also be applied to
insert the clusters of one tree into the other.  For both of these methods,
an overall linear-time algorithm for the common refinement of $k$ input
trees is then obtained by iteratively computing the common refinement of
the input tree $T_j$ and the common refinement of first $j-1$ trees,
resulting in a total effort of $O(k |L|)$.

Here we describe an alternative algorithm that constructs in a single step
a candidate refinement $T$ of all $k$ input trees. This is achieved by
computing the parent-function of the potential refinement $T$ in a
bottom-up fashion. As we shall see, the algorithm is easy to implement and
does not require elaborate data structures. The existence of a common
refinement is then verified by checking that the parent function defines a
tree $T$ and, if so, that $T$ displays each of the input trees $T_j$. This
test is also much simpler to implement than the isomorphism test for rooted
trees \cite{Aho:74}.

\section*{Theory}

\subsection*{Notation and Preliminaries}

Let $T$ be a rooted tree. We write $V(T)$ for its vertex set, $E(T)$ for is
edge set, $L(T)\subseteq V(T)$ for its leaf set,
$V^0(T)\coloneqq V(T)\setminus L(T)$ for the set of inner vertices and
$\rho\in V^0(T)$ for its root. An edge $e=\{u,v\}\in E(T)$ is an
\emph{inner} edge if $u,v\in V^0(T)$.  The ancestor partial order
$\preceq_T$ on $V(T)$ is defined by $x\preceq_T y$ whenever $y$ lies along
the unique path connecting $x$ and the root $\rho$. If $x\preceq_T y$ and
$x\ne y$, we write $x \prec_T y$.  For $v\in V(T)$, we set
$\child_T(v)\coloneqq\{u\mid\{v,u\}\in E(T),\, u\prec_T v\}$. If
$u\in\child_T(v)$, then $v$ is the unique parent of $u$. In this case, we write
$v=\parent_T(u)$. All trees $T$ considered in this contribution are
\emph{phylogenetic}, i.e., they satisfy $|\child_T(v)|\ge 2$ for all
$v\in V^0(T)$.

We denote by $T(u)$ the subtree of $T$ rooted in $u$ and write $L(T(u))$
for its leaf set. The \emph{last common ancestor} of a vertex set
$W\subseteq V(T)$ is the unique $\preceq_T$-minimal vertex
$\lca_T(W)\in V(T)$ satisfying $w\preceq_T\lca_T(W)$ for all $w\in W$. For
brevity, we write $\lca_T(x,y)\coloneqq\lca_T(\{x,y\})$. Furthermore, we
will sometimes write $vu\in E(T)$ as a shorthand for ``$\{u,v\}\in E(T)$
with $u\prec_T v$.'' 

A hierarchy on $L$ is set system $\HH\subseteq 2^L$ such that (i)
$L\in\HH$, (ii) $A\cap B\in\{A,B,\emptyset\}$ for all $A,B\in\HH$, and
(iii) $\{x\}\in\HH$ for all $x\in L$. There is a well-known bijection
between rooted phylogenetic trees $T$ with leaf set $L$ and hierarchies on
$L$, see e.g.\ \cite[Thm.\ 3.5.2]{Semple:03}. It is given by
$\HH(T) \coloneqq \{ L(T(u)) \mid u\in V(T) \}$; conversely, the tree
$T_{\HH}$ corresponding to a hierarchy $\HH$ is the Hasse diagram w.r.t.\
set inclusion. Thus, if $v=\lca_T(A)$ for some $A\subseteq L(T)$, then
$L(T(v))$ is the inclusion-minimal cluster in $\HH(T)$ that contains $A$,
see e.g.\ \cite{Hellmuth:21q}.  We call the elements of $\HH(T)$
\emph{clusters} and say that two clusters $C$ and $C'$ are
\emph{compatible} if $C\cap C'\in\{C,C',\emptyset\}$.  Note that, by (i),
the clusters of the same tree are all pairwise compatible.

A (rooted) triple is a binary tree on three leaves. We say that a tree $T$
displays a triple $xy|z$ if $\lca_T(x,y)\prec_T\lca_T(x,z)=\lca_T(y,z)$, or
equivalently, if there is a cluster $C\in \HH(T)$ such that $x,y \in C$ and
$z\notin C$. The set of all triples that are displayed by $T$ is denoted by
$r(T)$. A set $\mathcal{R}$ of triples is \emph{consistent} if there is a
tree that displays all triples in $\mathcal{R}$.

Let $T$ and $T^*$ be phylogenetic trees with $L(T)=L(T^*)$. We say that
$T^*$ is a \emph{refinement} of $T$ if $T$ can be obtained from $T^*$ by
contracting a subset of inner edges. Equivalently, $T^*$ is a refinement of
$T$ if and only if $\HH(T)\subseteq\HH(T^*)$. A tree $T$ \emph{displays} a
tree $T'$ if $L(T')\subseteq L(T)$ and
$\HH(T')\subseteq \{C\cap L(T') \mid C\in\HH(T) \text{ and } C\cap
L(T')\ne\emptyset\}$.  In particular, therefore, $T$ displays a tree $T'$
with $L(T')=L(T)$ if and only if $\HH(T')\subseteq\HH(T)$, i.e., if and
only if $T$ is a refinement of $T'$. The minimal \emph{common refinement}
of the trees $T_i$, $1\le i\le k$ is the tree $T$ such that
$\HH(T)=\bigcup_{i=1}^k \HH(T_i)$, provided it exists.

Thm.~3.5.2 of \cite{Semple:03} can be rephrased in the following form: 
\begin{lemma}
  \label{lem:leastresolved}
  Let $T_1$, $T_2$, \dots, $T_k$ be trees with common leaf set $L(T_i)=L$
  such that $\HH\coloneqq \bigcup_{i=1}^k\HH(T_i)$ is a hierarchy. Then
  there is a unique tree $T$ such that $\HH(T)=\HH$. Furthermore, $T$ is
  the unique ``least resolved'' tree in the sense that contraction of any
  edge in $T$ yields a tree $T_e$ with $\HH(T_e)\subsetneq \HH(T)$.
\end{lemma}
\begin{proof}
  By definition of $\HH$ and the bijection between phylogenetic trees
  and hierarchies, there is a unique tree $T$ such that
  $\HH=\HH(T)$. Consider an inner edge $e=uv$. By construction, there is
  at least one tree $T_v$ such that $C\coloneqq
  L(T(v))\in\HH(T_v)$. However, $\HH(T_e)=\HH(T)\setminus\{C_v\}$ and
  thus $T_e$ does not display $T_v$.
\end{proof}

By Thm.~1 in \cite{Bryant:95}, a tree $T'$ is displayed by a tree $T$ with
$L(T')\subseteq L(T)$ if and only if $r(T')\subseteq r(T)$.  As an
immediate consequence, a common refinement of trees with a common leaf set
$L$ exists if and only if the union $L$ of their triple sets is consistent.
The latter condition can be checked using the \texttt{BUILD} algorithm
which, in the positive case, returns a tree $\build(\mathcal{R},L)$ that
displays all triples in $\mathcal{R}$.

\begin{lemma}
  \label{lem:Aho}
  Suppose that $T$ is the unique least resolved common refinement of the
  trees $T_1$, $T_2$, \dots, $T_k$ with common leaf set $L(T_i)=L$,
  $1\le i\le k$ and let $\mathcal{R}\coloneqq r(T_i)\cup r(T_2)\cup \dots
  \cup r(T_k)$. Then $T=\build(\mathcal{R},L)$.
\end{lemma}
\begin{proof}
  The tree $\widehat{T}\coloneqq \build(\mathcal{R},L)$ is a common
  refinement since, by the arguments above, it displays $T_1$, $T_2$,
  \dots, $T_k$. By Lemma~\ref{lem:leastresolved}, we therefore have
  $\HH(T)\subseteq\HH(\widehat{T})$.  Prop.~4.1 in \cite{Semple:03a}
  implies that $\widehat{T}$ is least resolved w.r.t.\ $\mathcal{R}$, i.e.,
  every tree $\widehat{T}'$ obtained from $\widehat{T}$ by contraction of
  an edge no longer displays all input triples in $\mathcal{R}$.  By
  Thm.~6.4.1 in \cite{Semple:03}, $T_i$ is displayed by $\widehat{T}'$ if
  and only if $\widehat{T}'$ displays all triples of $T_i$. Since this is
  not true for all input trees $T_i$, $\widehat{T}'$ does not display all
  input trees $T_i$, $1\le i\le k$.  Together with
  $\HH(T)\subseteq\HH(\widehat(T))$, this implies that $T=\widehat{T}$.
\end{proof} 

We note that, given a set of triples $R$, ``$T$ is a least resolved
displaying $R$'' does not imply that vertex set $V(T)$ is minimal among all
such trees.  It is possible in general that there is a tree $T'$ displaying
a given triple set $R$ with $|V(T')|<|V(\build(R, L))|$.  In this case,
$\build(R,L)$ does not display $T'$, see \cite{Jansson:12} for
details. However, uniqueness of the least resolved tree,
Lemma~\ref{lem:leastresolved}, rules out this scenario in our setting.

The algorithm \texttt{BuildST} \cite{Deng:18} computes the supertree of a
set $\mathcal{T}\coloneqq \{T_i| 1\le i\le k\}$ of root trees without first
breaking down each tree to its triple set $r(T_i)$. Lemma~5 in
\cite{Deng:18} establishes that \texttt{BuildST} applied to a set of trees
and \texttt{BUILD} applied to the triple set
$\mathcal{R}\coloneqq\bigcup_{i=1}^k r(T_i)$ produce the same output for
all instances. If $\mathcal{R}$ is consistent, \texttt{BuildST} computed
the tree $\build(\mathcal{R},L)$. If all input trees have the same same
leaf set $L$ \texttt{BuildST} in particular computes their common
refinement. The performance analysis in \cite{Deng:18} shows that
\texttt{BuildST} runs in $O(k|L| \log^2(k|L|))$ time for this special case.
Linear-time algorithms for the special case of a common leaf set
therefore offer a further improvement over the best known general purpose
supertree algorithms.

\subsection*{A Bottom-Up Linear Time Algorithm}

The basic idea of our approach is to construct $T$ by means of a
simple bottom-up approach that computes the parent function
$\parent_T:V(T)\setminus\{\rho_T\}\to V(T)\setminus L$ of a candidate
tree $T$ in a stepwise manner. This algorithm is based on three simple
observations:
\begin{itemize}
  \item[(i)] If it exists, the common refinement $T$ of $T_1$, $T_2$, \dots,
  $T_k$ is uniquely defined by virtue of $\HH(T)=\bigcup_{i=1}^k \HH(T_i)$
  (cf.\ Lemma \ref{lem:leastresolved}). We will therefore identify all
  vertices $v_i\in V(T_i)$ with a vertex $v$ in the prospective tree $T$
  whenever their cluster, i.e., the sets $L(T_i(v_i))$ are the same. In
  this case, we have $L(T(v))\coloneqq L(T_i(v_i))$.  From here on, we
  simply say, by a slight abuse of notation, that $v$ is also a vertex of
  $T_i$ and write $v\in V(T_i)$.
  \item[(ii)] Since $\HH(T)=\bigcup_{i=1}^k \HH(T_i)$, each vertex
  $v\in V(T)$ is also a vertex in at least one input tree $T_i$.
  Conversely, every vertex $v\in V(T_i)$, $i\in\{1,\dots,k\}$, is a vertex
  in $T$. Therefore, we have $V(T)=\bigcup_{i=1}^k V(T_i)$.
  \item[(iii)] $T$ exists if and only if the sets $L(T(x))$ and $L(T(y))$ for
  all $x,y\in \bigcup_{i=1}^k V(T_k)$ are either comparable by set
  inclusion or disjoint, i.e.,
  $L(T(x))\cap L(T(y))\in\{L(T(x)),L(T(y)),\emptyset\}$.  Thus,
  $x \prec_T y$ if and only if
  $L(T_i(x))=L(T(x)) \subsetneq L(T_j(y))=L(T(y))$ for the appropriate
  choices of $i,j\in\{1,\dots, k\}$.
\end{itemize} Observation (iii) makes it possible to access the ancestor
order $\prec_T$ on $V(T)$ without knowing the common refinement $T$
explicitly.  Many of the upcoming definitions are illustrated in
Fig.~\ref{fig:exmpl-defs}.

\begin{figure}[t]
  \begin{center}
    \includegraphics[width=0.85\textwidth]{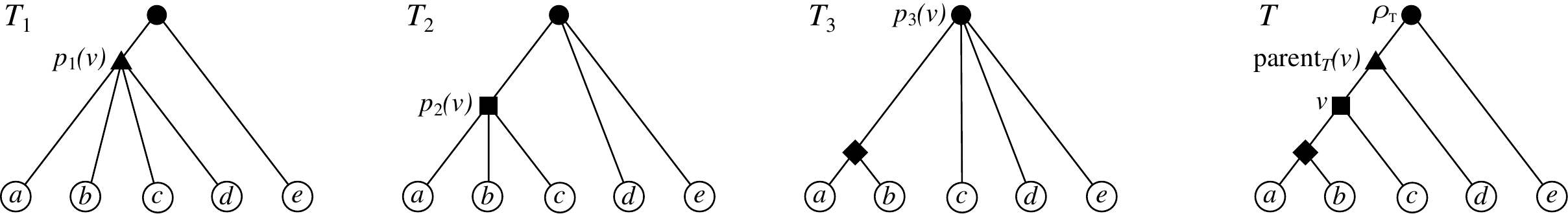}
  \end{center}
  \caption{The three trees $T_1$, $T_2$, and $T_3$ with common leaf set
    $L=\{a,b,c,d,e\}$ have the (unique) common refinement $T$. Here,
    $J(\rho)=\{1,2,3\}$ and thus, $\bar J(\rho) = \emptyset$.  The
    different symbols for vertices indicate which vertex $u$ in the $T_i$s
    corresponds to which vertex $u$ in $T$. Consider the vertex $v$
    highlighted as $\blacksquare$.  The corresponding vertices $p_i(v)$ are
    shown in the respective trees $T_i$.  Here, $p_2(v)=v$ while the
    vertices $p_1(v)$ and $p_3(v)$ in $T_1$ and $T_3$ correspond to
    $\parent_T(v)$ and $\rho$, respectively. Consequently, $J(v) = \{2\}$
    and $\bar J(v)=\{1,3\}$. We have
    $p_2(v)=v\prec_T\parent_T(v) = p_1(v)\prec_T p_3(b)$, according to
    Obs.\ \ref{fact:dicho}. In this example, only the last case in Obs.\
    \ref{obs:parent} for $v$ is satisfied, namely $\parent_T(v)=p_1(v)$.
    Moreover,
    $A(v) = \{v\} \cup \{\parent_{T_2}(v)=\rho\} \cup \{p_1(v),p_3(v)\} =
    \{v,\rho, \parent_T(v)\}$.   }
  \label{fig:exmpl-defs}
\end{figure}

We introduce, for each $v\in V(T)$, the index set
$J(v)=\{i\mid L(T_i(v))=L(T(v))\}$ of the trees that contain vertex $v$. We
have $J(v)\neq \emptyset$ for all $v\in V(T)$. For simplicity, we write
$\bar J(v)\coloneqq \{1,\dots,k\}\setminus J(v)$ for the indices of all
other trees. Hence, $\bar J(v)=\emptyset$ if and only if
$L(T(v))\in \HH(T_i)$ for all $i\in \{1,\dots,k\}$. In particular,
therefore, $\bar J(v)=\emptyset$ whenever $v\in L$ or $v=\rho$.

\emph{Let us assume until further notice that a common refinement exists
  and let $T=(V,E)$ be the unique least resolved common refinement of
  $T_1$, $T_2$, \dots, $T_k$ on a common leaf set.} Due to Lemma
\ref{lem:leastresolved}, $T$ is uniquely determined by the parent function
$\parent_T$.  The key ingredient in our construction are the following
vertices in $T_i$:
\begin{equation}
  \label{eq:pi}
  p_i(v)\coloneqq \lca_{T_i}(L(T(v)),\qquad i\in\{1,\dots,k\},\ v\in V(T)
\end{equation}
By assumption, we have $L(T(v))\subseteq L(T_i)$ and thus $p_i(v)$ is
well-defined.  As immediate consequence of the definition in
Eq.(\ref{eq:pi}), we have
\begin{fact}
  \label{fact:dicho}
  For all $v\in V(T)$ and all $i\in\{1,\dots,k\}$ it holds that $p_i(v)=v$
  iff $v\in V(T_i)$ iff $i\in J(v)$. If $i\notin J(v)$, then
  $v\prec_T p_{i}(v)$ and therefore $\parent_T(v)\preceq_T p_i(v)$.
\end{fact}
Now assume that $\parent_T(v)$ exists in $T$, i.e., $v\neq \rho$.  By
Observation (ii), $v\in V(T)$ implies $v\in V(T_i)$ for some
$i\in\{1,\dots,k\}$. In this case, $\parent_T(v)$ must be the unique
$\preceq_{T_i}$-minimal vertex $u_i\in V(T_i)$ that satisfies
$L(T(v))\subsetneq L(T_i(u_i))$ because $\HH(T_i)\subseteq \HH(T)$.  In
other words, $p_i(\parent_T(v)) = u_i = \parent_{T_i}(v)$.  Hence, we have
\begin{fact}\label{obs:parent}
  For all $v\in V\setminus\{\rho\}$ it holds that
  $\parent_T(v)=\parent_{T_i}(v)$ for some $i\in J(v)$ or
  $\parent_T(v)=p_{j}(v)$ for some $j\in \bar J(v)$.
\end{fact}
Note that in general also both cases in Obs.~\ref{obs:parent} are possible.
Consider the set of vertices
$A(v)\coloneqq\{v\}\cup\{\parent_{T_i}(v)\mid i\in J(v)\}\cup\{p_i(v)\mid
i\in\bar J(v)\}$. By construction and Obs.\ \ref{obs:parent}, we have
$v\preceq_T x$ for all $x\in A(v)$.  Since all ancestors of a vertex in a
tree are mutually comparable w.r.t.\ the ancestor order, we have
\begin{fact}
  \label{fact:comp}
  All $x,y\in A(v)$ are pairwise comparable w.r.t.\ $\preceq_T$.
\end{fact}
Taken together, Observations~\ref{fact:dicho}-\ref{fact:comp} imply that
the parent map of $T$ can be expressed in the following form:
\begin{equation}
  \label{eq:mini}
  \parent_T(v) =
  \min\left( \min_{i\in J(v)} \parent_{T_i}(v),\ \min_{i\in\bar J(v)}
  p_i(v) \right)
\end{equation}
where the minimum is taken w.r.t.\ the ancestor order $\preceq_T$ on $T$.
Since the root $\rho_i$ of each $T_i$ coincides with the root $\rho$ of
$T$, $v$ is the root of $T$ iff $\parent_{T_i}(v)=\varnothing$ is undefined
for one and thus for all $i$. In this case, we set
$\parent_T(v)=\varnothing$.

With this in hand, we show how to compute the maps $p_i$ for
$u\coloneqq \parent_T(v)$ for all $i\in \{1,\dots,k\}$. To this end, we
distinguish three cases. (1) If $u\in V(T_i)$, we have $p_i(u)=u$ by
definition. (2) If $u\notin V(T_i)$, then we have to identify the
$\preceq_T$-minimal vertex $w\in V(T_i)$ with $u\prec_T w$.  If
$v\in V(T_i)$, then $p_i(u)=w=\parent_{T_i}(v)$. In the remaining case,
$i\in \bar J(v)$, we already know that $p_i(v)$ is the
$\preceq_{T_i}$-minimal ancestor of $v$. Thus, we have either
$p_i(v)=u=\parent_T(v)$, i.e., a sub-case of (1), or (3)
$u\preceq_T p_i(v)$ whenever $v\notin V(T_i)$ and $u\notin V(T_i)$. In this
case, the definition of $p_i$ implies $p_i(u)=p_i(v)$. Summarizing the
three cases yields the following recursion:
\begin{equation}
  \label{eq:rec-pi}
  p_i(u) = \begin{cases}
    u                 & \text{ if } i\in J(u) \\
    \parent_{T_i}(v)  & \text{ if } i\in J(v) \\
    p_i(v)            & \text{ if } i\in \bar J(u)
    \text{ and } i \in\bar J(v)
  \end{cases}
\end{equation}
Note, although the cases in Eq.\ \eqref{eq:rec-pi} are not exclusive (since
$J(v)\cap J(u)\neq \emptyset$ is possible), they are not in conflict. To
see this, observe that if $i\in J(u)$ and $i\in J(v)$, then
$u=\parent_{T_i}(v)$ as a consequence of the definition of
$u$.

Initializing $i\in J(v)$ for all $i$ and all leaves $v$, we can
compute $J(u)$ for $u=\parent_T(v)$ as a by-product by the minimum
computation in Eq.(\ref{eq:mini}) by simply keeping track of the equalities
encountered since both $\parent_{T_i}(v)$ and $p_i(v)$ are vertices in
$T_i$. More precisely, each time a strictly $\preceq_T$-smaller vertex
$u'$, i.e., a proper set inclusion, is encountered in Eq.(\ref{eq:mini}),
the current list of equalities is discarded and re-initialized as $\{i\}$,
where $i$ is the index of the tree $T_i$ in which the new minimum $u'$ was
found. The indices of the trees $T_j$ with $u'\in V(T_j)$ are then
appended.

It remains to ensure that the vertices are processed in the correct order.
To this end, we use a queue $\mathcal{Q}$, which is initialized by
enqueueing the leaf set. Upon dequeueing $v$, its parent $u$ and the values
$p_i(u)$ are computed. Except for the leafs, every vertex $u\in V(T)$
appears as parent of some $v\in V(T)$. On the other hand, $u$ may appear
multiple times as parent. Thus we enqueue $u$ in $\mathcal{Q}$ only if the
same vertex has not been enqueued already in a previous step. We emphasize
that it is not sufficient to check whether $u\in\mathcal{Q}$ since $u$ may
have already been dequeued from $\mathcal{Q}$ before re-appearance as a
parent. We therefore keep track of all vertices that have ever been
enqueued in a set $V$. To see that this is indeed necessary, consider a
tree $T_i=(a,(b,c)v_1)v_2$ and an initial queue
$\mathcal{Q}=(a,b,c)$. Without the auxiliary set $V$, we obtain
$\mathcal{Q}=(b,c,v_2)$, $\mathcal{Q}=(c,v_2,v_1)$,
$\mathcal{Q}=(v_2,v_1)$, $\mathcal{Q}=(v_1)$, $\mathcal{Q}=(v_2)$,
\textit{etc.}, and thus $v_2$ is enqueued twice. 

An implementation of this procedure also needs to keep track of the
correspondence between vertices in $V(T)$ and the vertices of $V(T_i)$. To
this end, we can associate with each $v\in V(T)$ a list of pointers to
$v\in V(T_i)$ for $i\in J(v)$, and pointer from $v\in V(T_i)$ back to
$v\in V(T)$.  For the leaves, these are assigned upon
initialization. Afterwards, they are obtained for $u=\parent_{T}(v)$ as a
by-product of computing $J(u)$, since the pointers have to be set exactly
for the $i\in J(u)$.  In particular, whenever the pointer for $u$ found
$T_i$ has already been set, we know that $u\in V$.

Summarizing the discussion so far, we have shown:
\begin{proposition}
  Suppose the trees $T_1$, $T_2$, \dots, $T_k$ have a common refinement
  $T$. Then $\parent_{T}(v)$ is correctly computed by the recursions
  Eq.(\ref{eq:mini}) and Eq.(\ref{eq:rec-pi}).
\end{proposition}

Next we observe that it is not necessary to explicitly compute set
inclusions. As an immediate consequence of Obs.~\ref{fact:comp} and the
fact that $x\ne y$ implies $L(T(x))\ne L(T(y))$ because all trees are
phylogenetic by assumption, we obtain
\begin{fact}
  \label{fact:cardi}
  For any two $x,y\in A(v)$, we have $x\prec_T y$ if and only if
  $|L(T(x))|<|L(T(y))|$. 
\end{fact}
Thus it suffices to evaluate the minimum in Eq.(\ref{eq:mini}) w.r.t.\ to
the cardinalities $|L(T(v))|$. This can be achieved in $O(k)$ time provided
the values $\ell_i(v)\coloneqq |L(T_i(v))|$ are known for the input
trees. Since the parent-function $\parent_T$ unambiguously defines a tree
$T$, we have
\begin{corollary}
  Suppose the trees $T_1$, $T_2$, \dots, $T_k$ have a common refinement
  $T$. Then $T$ can be computed in $O(k|L|)$ time.
  \label{cor:linear-if-T-exists}
\end{corollary}
\begin{proof}
  For each input tree $T_i$, $\ell_i(v)$ can be computed as
  \begin{equation}
    \label{eq:ell-i}
    \ell_i(v)=
    \begin{cases}
      1              & \text{if } v\in L, \text{ and} \\
      \ell_i(v)=\displaystyle\sum_{u\in\child_{T_i}(v)}\ell_i(u) 
      & \text{otherwise.}
    \end{cases}
  \end{equation}
  Since the total number of terms appearing for the inner vertices of $T$
  equals the number of edges of $T_i$, the total effort for $T_i$ is
  bounded by $O(|L|)$. The total number of vertices $u$ computed as
  $\parent_{T}(v)$ equals the number of edges of $T$, and thus is also
  bounded by $O(L)$.  Since the tree $T$, as well as the $k$ trees $T_i$,
  have $O(|L|)$ vertices, we require $O(k|L|)$ pointers from the vertices
  in $T$ to their corresponding vertices in the $T_i$ and \textit{vice
    versa}.  By initializing the pointers for all $v\in V(T_i)$ as ``not
  set'', it can be checked in constant time whether $u$ that was found in
  $T_i$ is already contained in the set $V$, since this is the case if and
  only if its pointer has already been set.  Evaluation of
  Eq.(\ref{eq:mini}) requires $O(k)$ comparisons, which, by virtue of
  Obs.~\ref{fact:cardi}, can be performed in constant time. The computation
  of $p_i(u)$ and $J(u)$ as well as the update of the correspondence table
  between vertices in $T$ and $T_i$, $1\le i\le k$ requires $O(k)$
  operations for each $v\in V(T)$. Thus $T$ can be computed in $O(k|L|)$
  time.
\end{proof}

So far, we have assumed that a common refinement exists.  By a slight abuse
of notation, we also use the function $\parent_T$ if the refinement $T$
does not exist. In this case, we define $\parent_T$ on the union of the
$V(T_i)$ recursively by Eqs.(\ref{eq:mini}) and (\ref{eq:rec-pi}).
Alg.~\ref{alg:refinement-k-trees} summarizes the procedure based on the
leaf set cardinalities for the general case.  If no common refinement $T$
exists, then either $\parent_T$ does not specify a tree, or the tree $T$
defined by $\parent_T$ is not a common refinement of $T_1$, $T_2$, \dots,
$T_k$. The following result shows that we can always efficiently compute
$\parent_T$ and check whether it specifies a common refinement of the input
trees.

\begin{algorithm}[t]
  \caption{\texttt{LinCR} Common refinement for $k$ trees on the same 
    leaf set.}
  \label{alg:refinement-k-trees}
  
  \newcommand{\treeline}[1]{\begingroup\color{blue}#1\endgroup}
  
  \DontPrintSemicolon
  \SetKwFunction{FRecurs}{void FnRecursive}%
  \SetKwFunction{FRecurs}{Edit}
  \SetKwProg{Fn}{Function}{}{}
  
  \KwIn{Trees $T_1, T_2, \dots, T_k$, $k\ge 2$, with
    $L\coloneqq L(T_1)=\dots=L(T_k)$ and $|L|\ge 2$.}  \KwOut{Common
    refinement $T$ of $T_1, \dots, T_k$ if it exists, \textbf{false}
    otherwise.}
  
  \BlankLine
  compute $\ell_i(v)$ for all $v\in T_i$ and all $T_i$ according to 
  Eq.~(\ref{eq:ell-i}) \label{line:ell-i} \;
  initialize empty queue $Q$ and a set $V\leftarrow L$\;
  \BlankLine
  \ForEach{$v\in L$}{
    \textbf{enqueue} $v$ to $Q$\;
    $J(v)=\{1,\dots,k\}$;
    $p_i(v)\leftarrow v$ for all $i\in\{1,\dots k\}$;
    $\ell(v)=1$\;
  }
  
  \BlankLine
  
  \While{$Q$ is not empty}{
    $v\leftarrow$ \textbf{dequeue} first element from $Q$\;
    $u\leftarrow \varnothing$;
    $\ell_{\min}\leftarrow |L|$;
    $J_u\leftarrow \varnothing$\;
    \ForEach{$i\in \{1,\dots,k\}$ \label{line:min-start}}{
      \lIf{$i\in J(v)$}{
        $u'\leftarrow \parent_{T_i}(v)$ }
      \lElse{$u'\leftarrow p_i(v)$}
      \uIf{$u=\varnothing$ or $\ell_i(u')<\ell_{\min}$}{
        $u\leftarrow u'$;
        $\ell_{\min}\leftarrow \ell_i(u')$;
        $J_u\leftarrow \{i\}$\;
      }
      \lElseIf{$\ell_i(u')=\ell_{\min}$}{
        $J_u\leftarrow J_u\cup\{i\}$ \label{line:min-end}
      }
    }
    \lIf{$\ell(v)<\ell_{\min}$ \label{line:v-lt-u-check} }{
      $\parent_{T}(v)\leftarrow u$
    }
    \lElse 
    {
      \Return \textbf{false} \label{line:exit-v-ge-u}
    }
    
    \If{$u\notin V$ and $\ell_{\min}< |L|$ \label{line:enqueue-u}}{
      \textbf{enqueue} $u$ to $Q$ and add $u$ to $V$\;
      \lIf{$|V|>2|L|-2$}{
        \Return \textbf{false} \label{line:exit-too-many}
      }
      $\ell(u)=\ell_{\min}$;
      $J(u)\leftarrow J_u$\;
      \ForEach{$i\in\{1,\dots,k\}$}{
        \lIf{$i\in J(u)$}{
          $p_i(u)\leftarrow u$
        }
        \lElseIf{$i\in J(v)$}{
          $p_i(u)\leftarrow \parent_{T_i}(v)$
        }
        \lElse{
          $p_i(u)\leftarrow p_i(v)$ \label{line:pi-case-3} 
        }
      }
    }
  }
  $T\leftarrow$ the tree defined by the map $\parent_T$ 
  \label{line:T-from-parent-T}\;
  \lIf{$T$ is not phylogenetic \label{line:T-phylogenetic} }{
    \Return \textbf{false}
  }
  \ForEach{$i\in \{1,\dots,k\}$}{
    initialize a copy $T'_i$ of $T$\;
    \ForEach{$v\in V(T'_i)$ such that $i\in \bar J(v)$}{
      contract the edge $\{\parent_{T'_i}(v), v\}$ in $T'_i$\;      
    }
    \ForEach{$v\in V(T'_i)$ \label{line:isomorphism-chech-start} }{
      \lIf{$\child_{T'_i}(v)\ne \child_{T_i}(v)$}{
        \Return \textbf{false} \label{line:isomorphism-chech-end} 
      }
    }
  }
  \Return $T$
\end{algorithm}

\begin{theorem}
  \texttt{LinCR} (Alg.~\ref{alg:refinement-k-trees}) decides in $O(k|L|)$
  time whether a common refinement of trees $T_1$, $T_2$, \dots, $T_k$ on
  the same leaf set $L$ exists and, in the affirmative case, returns the
  tree $T$ corresponding to
  $\HH(T)=\HH(T_1)\cup\HH(T_2)\cup\dots\cup\HH(T_k)$.
\end{theorem}
\begin{proof}
  We construct $\parent_{T}$ in
  Lines~\ref{line:ell-i}--\ref{line:pi-case-3} as described in the proof of
  Cor.~\ref{cor:linear-if-T-exists}. In particular, we determine
  $u\coloneqq\parent_T(v)$ by virtue of the smallest $\ell_i(u)$.  Hence,
  we can process each enqueued vertex $v$ in $O(k)$.  Moreover, if a common
  refinement $T$ exists, then Cor.~\ref{cor:linear-if-T-exists} guarantees
  that we obtain this tree in Line~\ref{line:T-from-parent-T}.
  
  A tree on $|L|$ leaves has at most $|L|-1$ inner vertices with equality
  holding for binary trees.  Therefore, the set $V$ of distinct vertices
  encountered in Alg.~\ref{alg:refinement-k-trees}, can contain at most
  $2|L|-2$ vertices (note that by construction the root does not enter
  $V$).  If this condition is violated, no common refinement exists and we
  can terminate with a negative answer (cf.\
  Line~\ref{line:exit-too-many}). This ensures that $\parent_{T}$ is
  constructed in $O(k|L|)$ time.  We continue by showing that, unless the
  algorithm exits in Line~\ref{line:exit-v-ge-u}
  or~\ref{line:exit-too-many}, $\parent_{T}$ in
  Line~\ref{line:T-from-parent-T} always defines a tree $T$. To see this,
  consider the graph $G$ with vertex set $V\cup\{\rho\}$ where $\rho$ is
  the root vertex which is contained in each $T_i$ and an edge $\{u,v\}$ if
  and only if $\parent_{T}(v)=u$ or $\parent_{T}(u)=v$.  Checking whether
  $\ell(v)<\ell_{\min}(=\ell(u))$ in Line~\ref{line:v-lt-u-check} ensures
  that $G$ does not contain cycles and that $\parent_{T}(v)=u$ \emph{and}
  $\parent_{T}(u)=v$ is not possible.  Moreover, every vertex $v\in V$ is
  enqueued to $\mathcal{Q}$ and receives a parent $u$ such that
  $\ell(v)<\ell(u)$.  Unless $u=\rho$, $u$ in turn receives a parent $u'$
  with $\ell(u)<\ell(u')$.  Since $V$ is finite $v,u,u',...$ are pairwise
  distinct as a consequence of the cardinality condition, and we conclude
  that eventually $\rho$ is reached, i.e., a path to $\rho$ exists for all
  $v\in V$.  It follows that $G$ is connected, acyclic, and simple, and
  thus a tree (with root $\rho$).
  
  It remains to check whether $T$ is phylogenetic and displays $T_i$ for
  all $i\in\{1,\dots,k\}$.  Checking whether $T$ is phylogenetic in
  Line~\ref{line:T-phylogenetic} can be done in $O(|L|)$ in a top-down
  traversal that exits as soon as it encounters a vertex with a single
  child. To check whether $T$ displays a tree $T_i$, we contract (in a copy
  of $T$) in a top-down traversal all edges $uv$ with $v\in \child_T(u)$
  for which $u\notin V(T_i)$, i.e., for which $i\notin J(v)$.  Since the
  root of $T$ and leaves of $T$ are in $T_i$, this results in a rooted tree
  $T_i'$ with $V(T_i)=V(T_i')$ if $T$ is indeed the common refinement of
  all trees. The contraction of an edge $uv$ can be performed in
  $O(\child_T(v)|)$, hence in total time $O(|E(T_i)|)=O(|L|)$. Finally, we
  can check in $O(|L|)$ time whether the known correspondence between the
  vertices of $T_i$ and $T_i'$ is an isomorphism. To this end, it suffices
  to traverse $T_i$ and to check that $\child_{T_i}(v)=\child_{T'_i}(v)$
  for all $v\in V(T_i)$ (cf.\ Lines
  \ref{line:isomorphism-chech-start}--\ref{line:isomorphism-chech-end})
  using the pointers of $v$ and all elements in $\child_{T_i}(v)$ to the
  corresponding vertices in $T$. Note that, in general, the pointer from a
  vertex $v$ in $T_i$ to a vertex in $T'_i$ may not be set, in which case
  $v\notin V(T'_i)$ and thus, we can terminate with a negative answer.  The
  total effort thus is bounded by $O(k|L|)$.
  
  If $T$ on $L$ is a phylogenetic tree displaying all trees $T_1$, $T_2$,
  \dots, $T_k$, then it is a common refinement of these trees.  Since every
  vertex $v\in V(T)$ is also contained in some $T_i$, i.e.,
  $L(T(v))=L(T_i(v))$, we have
  $\HH(T)=\HH(T_1)\cup\HH(T_2)\cup\dots\cup\HH(T_k)$.
\end{proof}

\section*{Computational Results} 

We compare the running times for (\textsc{a}) \texttt{BUILD} \cite{Aho:81},
(\textsc{b}) \texttt{BuildST} \cite{Deng:18}, (\textsc{c})
\texttt{Merge_Trees} \cite{Jansson:16}, (\textsc{c'})
\texttt{Loose_Cons_Tree} \cite{Jansson:16}, and (\textsc{d}) \texttt{LinCR}
(Alg.~\ref{alg:refinement-k-trees}).  To this end, we implemented all of
these algorithms in Python as part of the \texttt{tralda} library.  We note
that \texttt{BUILD} operates on a set of triples extracted from the input
trees rather than the trees themselves.  We use the union of the minimum
cardinality sets of representative triples of every $T_i$ appearing in the
proof of Thm.~2.8 in \cite{Gruenewald:07a}.  Therefore, we have
$R\in O(k|L|^2)$ \cite[Thm.~6.4]{Seemann:18} and \texttt{BUILD} runs in
$O(k|L|^3)$ time.  In the case of \texttt{Merge_Trees}, we implemented a
variant that starts with $T=T_1$ and then iteratively merges the clusters
of the tress $T_i$, $2\le i\le k$, into $T$.  \texttt{Merge_Trees} assumes
that the input trees are compatible, which is guaranteed in our
benchmarking data set.  In practice, however, this condition may be
violated, in which case the behavior of \texttt{Merge_Trees} is undefined.
We therefore also implemented an $O(k|L|)$ algorithm for constructing the
\emph{loose consensus tree} for a set of trees $T_1$, $T_2$, \dots, $T_k$
on the same leaf set, \texttt{Loose_Cons_Tree}, following
\cite{Jansson:16}.  The loose consensus comprises all clusters that occur
in at least one tree $T_i$, $1\le i\le k$ and that are compatible with all
other clusters of the input trees (see \cite{Bremer:90, Day:03,Dong:11} and
the references therein).  The loose consensus tree by definition coincides
with the common refinement whenever the latter exists.
\texttt{Loose_Cons_Tree} uses \texttt{Merge_Trees} as a subroutine but
ensures compatibility in each step by first deleting incompatible clusters
in one of the trees. This is implemented as the deletion of the
corresponding inner vertex $v$ followed by reconnecting the children of $v$
to the parent of $v$. The input trees are compatible if and only if no
deletion is necessary. The existence of a common refinement can therefore
by checked by keeping track of the number of deletions.  However, the
subroutine that processes trees to remove incompatible clusters
significantly adds to the running time of the \texttt{Loose_Cons_Tree}
algorithm. The linear-time algorithms require $O(k|L|)$ space.

We simulate test instances as follows: First, a random tree $T^*$ is
generated recursively by starting from a single vertex (which becomes the
root) and stepwise attaching new leaves to a randomly chosen vertex $v$
until the desired number of leaves $|L|$ is reached.  In each step, we add
two children to $v$ if $v$ is currently a leaf, and only a single new leaf
otherwise.  This way, the number of leaves increases by exactly one in each
step and the resulting tree $T^*$ is phylogenetic (but in general not
binary).  From $T^*$, we obtain $k\in\{2,8,32\}$ trees $T_1$, $T_2$,\dots,
$T_k$ by random contraction of inner edges in (a copy of) $T^*$.  Each edge
is considered for contraction independently with a probability
$p\in\{0.1,0.5,0.9\}$. Therefore, $T^*$ is a refinement of $T_i$ for all
$1\le i\le k$, i.e., a common refinement exists by construction. However,
in general we have $\HH(T^*)\ne \bigcup_{i=1}^k \HH(T_i)$, i.e., $T^*$ is
not necessarily the minimal common refinement of the $T_i$.  The trees
$T_1$, $T_2$, \dots, $T_k$ constructed in this manner serve as input for
all algorithms.

\begin{figure}[t]
  \begin{center}
    \includegraphics[width=0.85\textwidth]{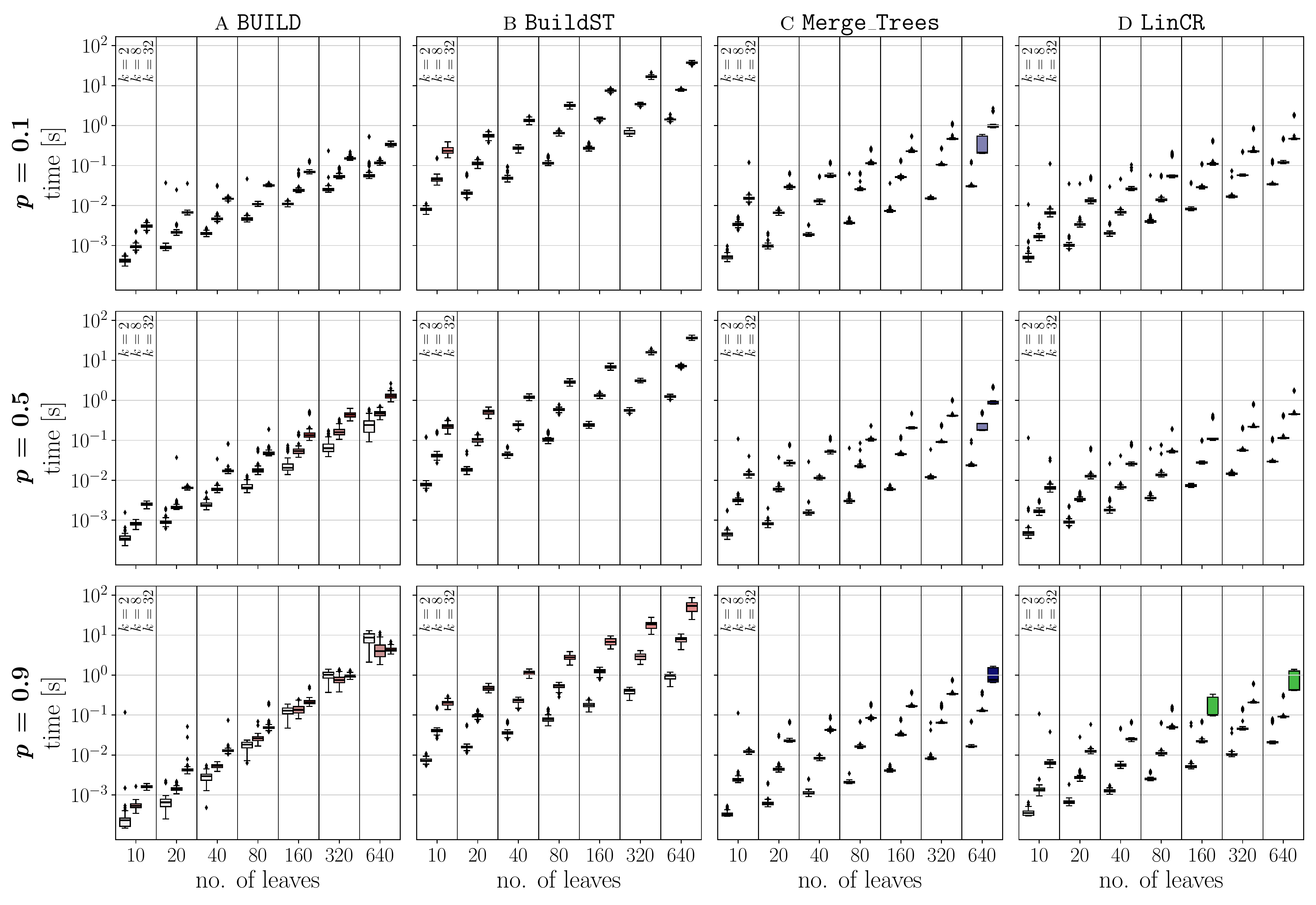}
  \end{center}
  \caption{Running time comparison of the algorithms for the construction
    of a common refinement of $k$ input trees on leaf set $L$.  The
    subplots of each row show boxplots for the running time for different
    numbers of leaves $|L|$ (indicated on the x-axis) and different values
    of $k\in\{2, 8, 32\}$ (indicated in the leftmost column of each
    subplot).  In each row, a different probability $p\in\{0.1, 0.5, 0.9\}$
    for edge contraction was used to produce the $k$ input trees.  Per
    combination of the parameters $|L|$, $k$, and $p$, 100 instances were
    simulated to which all four algorithm were applied.}
  \label{fig:runtime-boxplots}
\end{figure}

\begin{figure}[t]
  \begin{center}
    \includegraphics[width=0.7\textwidth]{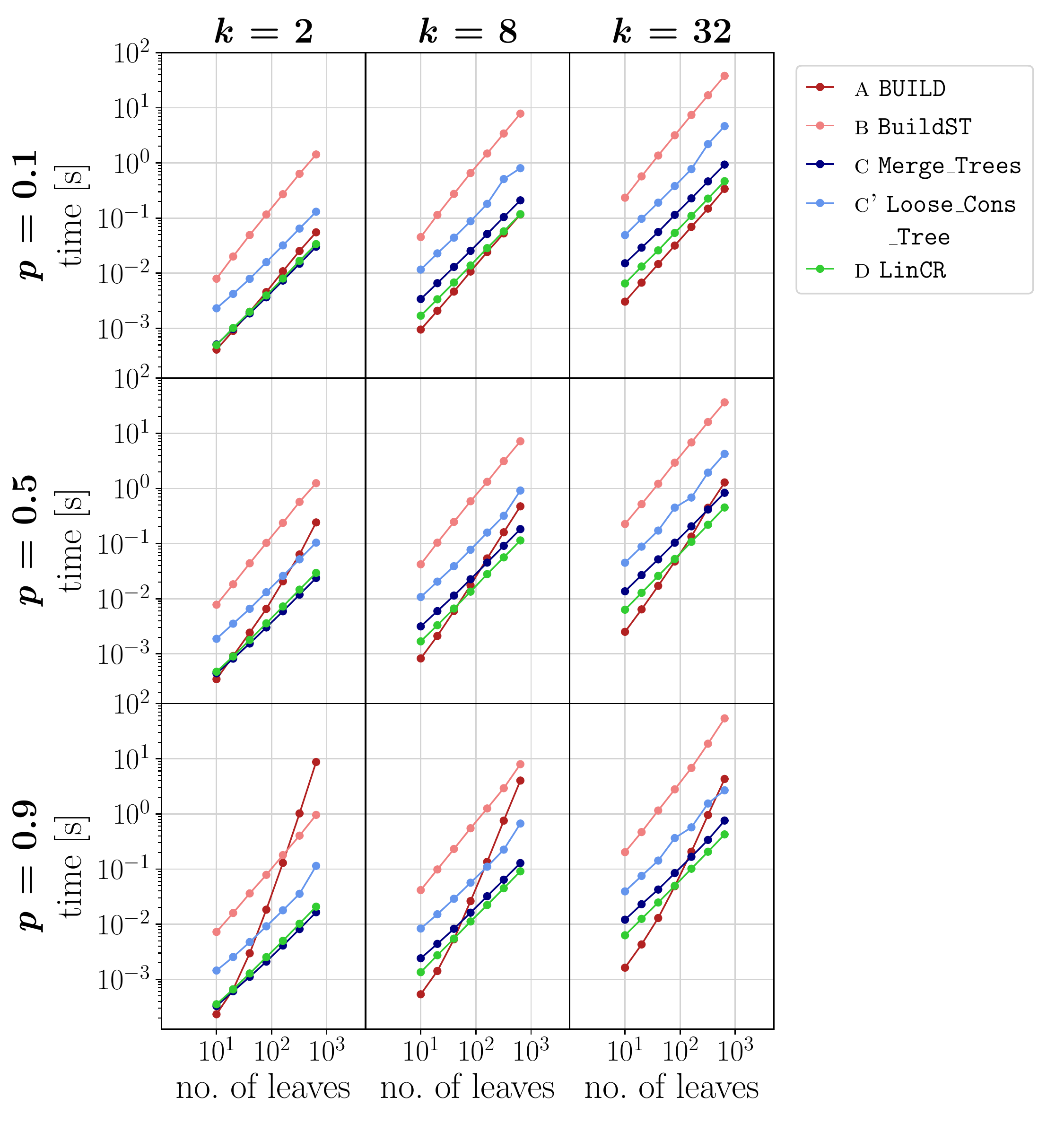}
  \end{center}
  \caption{Running time comparison of the algorithms for the construction
    of a common refinement of $k$ input trees on leaf set $L$.  Per
    combination of the parameters $|L|$ (indicated on the horizontal axis),
    $k$ (columns), and $p$ (rows), 100 instances were simulated and median
    values are shown for all algorithms.  In each row, a different
    probability $p\in\{0.1, 0.5, 0.9\}$ for edge contraction was used to
    produce the $k$ input trees.}
  \label{fig:runtime-medians}
\end{figure}

The running time comparisons were performed using \texttt{tralda} on an
off-the-shelf laptop (Intel\textsuperscript{\textregistered}
Core\texttrademark\;i7-4702MQ processor, 16~GB RAM, Ubuntu 20.04, Python
3.7).  The time required to compute a least resolved common refinement of
the input trees is included in the respective total running time shown in
Figs.~\ref{fig:runtime-boxplots} and~\ref{fig:runtime-medians}. The
empirical performance data are consistent with the theoretical result that
\texttt{LinCR} scales linearly in $k|L|$.  In particular, the median
running times scales linearly with $|L|$, as shown by the slopw of
$\approx 1$ in the log/log plot for the running times of \texttt{LinCR} in
Fig.~\ref{fig:runtime-medians}.

In accordance with the theoretical complexity of $O(k|L| \log^2(k|L|))$ for
the common refinement problem, the performance curve of \texttt{BuildST} is
almost parallel to that of \texttt{LinCR}; however, its computation cost is
higher by almost two orders of magnitude.  For both algorithms, the
contraction probability $p$ appears to have little effect on the running
time. In both cases, a larger value of $p$ (i.e., a lower average
resolution of the input trees) leads to a moderate decrease of the running
time.

In contrast, the resolution of the input trees has a large impact on the
efficiency of \texttt{BUILD}. It also scales nearly linearly when the
resolution of the individual input trees $T_i$ is comparably high (and even
terminates faster than \texttt{LinCR} up until a few hundred leaves, cf.\
top-right panel), whereas its performance drops drastically with increasing
$p$, i.e., for poorly resolved input trees.  The reason for this is most
likely the cardinality of a minimal triple set that represents the set of
input trees.  For binary trees, the cardinality of the triple set of $T_i$
equals the number of inner edges \cite{Gruenewald:07a}, i.e., there are
$O(|L|)$ triples. For very poorly resolved trees, on the other hand,
$O(|L|^2)$ triples are required \cite{Seemann:18}, matching the differences
of the slopes with $p$ observed for \texttt{BUILD} in
Fig.~\ref{fig:runtime-boxplots}.

As expected, the curves of the two $O(k|L|)$ algorithms
\texttt{Merge_Trees} and \texttt{Loose_Cons_Tree} are also almost parallel
to that of \texttt{LinCR}.  For $k=2$, we can even observe that
\texttt{Merge_Trees} is slightly faster than \texttt{LinCR}. However, the
smaller number of necessary tree traversals in \texttt{LinCR} apparently
becomes a noticeable advantage with an increasing number $k$ of input
trees.  The additional tree processing steps in the more practically
relevant \texttt{Loose_Cons_Tree} algorithm, furthermore, result in a
longer running time compared to our new approach.

\section*{Concluding Remarks} 

We developed a linear-time algorithm to compute the common refinement of
trees on the same leaf set.  In contrast to the ``classical'' supertree
algorithms \texttt{BUILD} and \texttt{BuildST}, \texttt{LinCR} uses a
bottom-up instead of a top-down strategy. This is similar to
\texttt{Loose_Cons_Tree} and its subroutine \texttt{Merge_Trees}
\cite{Jansson:16}, which can also be used to obtain the common refinement
of trees on the same leaf set in linear time.  \texttt{LinCR}, however,
requires fewer tree traversals and is, in our opinion, simpler to
implement. In contrast to \texttt{Merge_Trees}, \texttt{LinCR} in
particular does not rely on a data structure that enables linear-time tree
preprocessing and constant-time last common ancestor queries for the nodes
in the tree \cite{Bender:05}.  All algorithms were implemented in Python
and are freely available for download from
\url{https://github.com/david-schaller/tralda} as part of the
\texttt{tralda} library. Empirical comparisons of running times show that
\texttt{LinCR} consistently outperforms the linear-time alternatives.
Only \texttt{BUILD} is faster for very small instances and moderate-size
trees that are nearly binary.

Although it may be possible to improve Alg.~\ref{alg:refinement-k-trees} by
a constant factor, it is asymptotically optimal, since the input size is
$O(k|L|)$ for $k$ trees with $|L|$ leaves. Furthermore, trivial solutions
can be obtained in some limiting cases. For instance, if $|V(T_i)|=2|L|-1$,
then $T_i$ is binary, i.e., no further refinement is possible. In this case,
we can immediately use $T=T_i$ as the only viable candidate and only check
that $T_j$ displays all other $T_j$.  However, we cannot entirely omit
Lines~\ref{line:ell-i}--\ref{line:pi-case-3} in this case since we require
the sets $J(v)$ as well as the correspondence between the vertices in order
to check whether $T$ displays every $T_i$. 

It is worth noting that the idea behind \texttt{LinCR} does not generalize
to more general supertree problems. The main reason is that the set
inclusions employed to determine $\prec_T$ do not carry over to the more
general case because the inclusion order of $C_1,C_2\in\HH(T)$ cannot be
determined from $C_1\cap L(T_i)$ and $C_2\cap L(T_j)$ for two trees with
$L(T_i),L(T_j)\subsetneq L(T)$.

Depending on the application, a negative answer to the existence of a
common refinement may not be sufficient. One possibility is to resort to
the loose consensus tree or possibly other notions of consensus trees, see
e.g.\ \cite{Bremer:90,Bryant:03}. A natural alternative approach is to
extract a maximum subset of consistent triples from
$\bigcup_{i=1}^k r(T_i)$. This problem, however, is known to be NP-hard for
arbitrary triple sets, see e.g.\ \cite{Byrka:10} and the references
therein.

\bigskip\noindent
\textbf{Acknowledgements.} 
This work was supported in part by the German Research Foundation
(\emph{DFG}), proj.\ no.\ STA850/49-1.

\bigskip\noindent
\textbf{Availability of data and materials.}
Implementations of the algorithms used in this contribution are available
at \url{https://github.com/david-schaller/tralda} as part of the
\texttt{tralda} library.

\bibliography{CR-preprint2}

\end{document}